\RequirePackage{fix-cm}
\documentclass[smallextended]{svjour3}       % onecolumn (second format)
\smartqed  % flush right qed marks, e.g. at end of proof
\usepackage{graphicx}
\usepackage{enumerate}
\usepackage{float}
\usepackage{caption}
\usepackage[justification=centering]{caption}
\usepackage{amsmath}
\usepackage{comment}
\usepackage{mathrsfs}
\usepackage{color}
\usepackage[colorlinks,urlcolor=blue,citecolor=blue]{hyperref}
\usepackage{cite}
\usepackage{multirow}
\usepackage{colortbl}
\usepackage{algorithmicx,algorithm}
\usepackage{threeparttable}
\usepackage{multirow}
\definecolor{mygray}{gray}{.70}

\newcommand{\tabincell}[2]{\begin{tabular}{@{}#1@{}}#2\end{tabular}}

\begin{document}

\title{A verifiable framework of entanglement-free quantum secret sharing with information theoretical security%\thanks{Grants or other notes
}


%\author{First Author         \and
%        Second Author %etc.
%}

\author{Changbin Lu$^1$          \and
        Fuyou Miao$^1$          \and
        Junpeng Hou$^2$          \and
        Wenchao Huang$^1$  \and
        Yan Xiong$^1$
}

%\affiliation{$^1$School of Computer Science and Technology, University of Science and Technology of China, Hefei, China\\
%$^2$Department of Physics, The University of Texas at Dallas, Richardson, Texas 75080-3021, USA}
\institute{Changbin Lu \at
              %first address \\
              \email{lcb@mail.ustc.edu.cn}           %  \\
           \and
           Fuyou Miao \at
              \email{mfy@ustc.edu.cn}
           \and
           Junpeng Hou \at
             \email{Junpeng.Hou@utdallas.edu}
           \and
           Wenchao Huang \at
                \email{huangwc@ustc.edu.cn}
           \and
           Yan Xiong \at
            \email{yxiong@ustc.edu.cn}\\
           $^1$ School of Computer Science and Technology, University of Science and Technology of China, Hefei, China\\
          $^2$ Department of Physics, The University of Texas at Dallas, Richardson, Texas 75080-3021, USA
}
\date{Received: date / Accepted: date}

\maketitle

\begin{abstract}
Quantum secret sharing (QSS) schemes without entanglement have huge advantages in scalability and are easier to realize as they only require sequential communications of a single quantum system. However, these schemes often come with drawbacks such as exact ($n, n$) structure, security flaws and absences of effective cheating detections. To address these problems, we propose a verifiable framework by utilizing entanglement-free states to construct ($t, n$)-QSS schemes. Our work is the heuristic step towards information-theoretical security in entanglement-free QSS, and it sheds light on how to establish effective verification mechanism against cheating. As a result, the proposed framework has a significant importance in constructing QSS schemes for versatile applications in quantum networks due to its intrinsic scalability, flexibility and information theoretical security.
\keywords{Quantum cryptography \and Quantum secret sharing \and Entanglement-free \and Information theoretical security \and Verification mechanism}
\end{abstract}

\section{Introduction}

Sharing a secret among multiple users with efficiency is a significant problem in practice. Currently, several schemes called secret sharing (SS) are proposed to attain this goal. As an important class of SS schemes, ($t, n$) threshold SS [($t, n$)-SS] was proposed by Shamir \cite{Shamir1979} and Blakely \cite{Blakley1979} independently in $1979$ and it has been applied in various fields \cite{Harn2013,Boldyreva2002,Harn1994,Liu2016,Desmedt1994,Patel2016}. However, the security of classical cryptography systems usually relies on the assumptions of computational complexity, which can be easily weakened by the development of advanced computational tools such as quantum algorithms \cite{Shor1994,Grover1997}. So quantum cryptography has attracted much attention due to its inherent security, which is ensured by physical laws such as the resulted quantum no-cloning theorem \cite{Wootters1982,Dieks1982} and Heisenberg uncertainty principle. Due to those properties, using quantum secret sharing (QSS) (the description of abbreviations see Appendix.\ref{A}) to share secrets among users is more promising and reliable. Furthermore, it also provides a robust and secure solution to storage and manipulation of quantum states \cite{Cleve1999}. In $1999$ Hillery $et$ $al.$ (HBB) proposed the first QSS scheme \cite{Hillery1999}, which takes advantage of a three-qubit entangled Greenberger-Horne-Zeilinger (GHZ) state. In the scheme, a GHZ triplet is split and each particle is delivered to a user. Both users measure their own particles in natural basis and combine the results to obtain the dealer's measurement result. Following the similar idea, various (HBB-type) QSS schemes using quantum correlations in well-constructed entangled states are proposed \cite{Karlsson1999,Zhang2005,Gottesman2000,Fu2015,Marin2013,He2008}. However, these entanglement-based schemes are poor in scalability because engineering entanglement among large number of particles is unrealistic in current experiments \cite{Chen2005,Wang2016}. Moreover, quantum correlations can be degraded by decoherence which is caused by weak interactions with the environment and this may lead to undesired results \cite{Unruh1995}. Another problem in these schemes is that a failure measurement owing to inefficient detection can easily render an invalid round.

Thus, QSS schemes without entanglement are more practical to quantum network and have attracted a lot of attentions. For instance, in Schmid's scheme \cite{Schmid2005}, a single qubit is used to carry the secret which can be recovered with sequential phase shift operations. Later, a recent work \cite{Tavakoli2015} uses mutually unbiased (orthonormal)
bases (MUBs) to construct QSS scheme. In the paper, it shows that schemes based on multi-system entanglement can be mapped into much simpler ones involving only single system. Then, V. Karimipour $et$ $al.$ \cite{Karimipour2015} further optimized the method and improved the efficiency from $1/d$ to $1/2$. Recently, the papers \cite{Lu2018,Bai2018} enable the QSS scheme constructed using the MUBs of \cite{Tavakoli2015} to have a threshold or general access structure. Such schemes only require sequential communication of a single quantum state and thus, have significant advantages in scalability and make proof-of-concept experiments feasible.  However, they \cite{Schmid2005,Tavakoli2015,Karimipour2015} are all ($n, n$) structure which are not authentic threshold secret sharing in nature since all $n$ shareholders are required to participate in secret recovering. Besides, most of them have been shown to be vulnerable to attacks. For example, Schmid's scheme and Karimipour's scheme are no longer secure if someone adopts the attacks proposed in \cite{He2007,Lin2016}. Moreover, their schemes assume a trusted third party ($e.g.$ the dealer) in secret reconstruction, who is appointed to measure the processed quantum states; such a cheating is not considered in their schemes that a participant could announce fake random numbers in secret reconstruction to obstacle the other participants from recovering the real secret.

In this paper, we first generalize the method in previous works \cite{Schmid2005,Tavakoli2015,Karimipour2015,Lu2018,Bai2018} and  propose a verifiable framework for entanglement-free $(t,n)$-QSS schemes. Then a concrete implementation is given. Moreover, schemes constructed under the verifiable framework enjoy the following advantages over previous schemes.

($i$) \textit{Scalability}. These schemes only require sequential communication of a single quantum system, $i.e.$, each participant performs their own unitary operations sequentially on a quantum state. As a result, they can be easily applied to large number of participants.

($ii$) \textit{Threshold structure}. Though being a QSS, the exact $(t,n)$ threshold structure can be realized through incorporating versatile classical methods such as interpolation polynomial \cite{Shamir1979}, geometric structure \cite{Blakley1979}, linear code \cite{McEliece1981,Massey1993}, Chinese Remainder Theorem \cite{Asmuth1983,Mignotte1982}, etc.

($iii$) \textit{Verification mechanism}. The last participant (sequential order) can measure the final state in nature basis (easier than MUBs in \cite{Schmid2005,Tavakoli2015,Karimipour2015,Lu2018}) and this frees the requirement of trusted third parties. Thanks to the auxiliary particle, cheating of any participant and eavesdropping can be detected by verifying the consistency of the recovered results.

\section{Framework of entanglement-free ($t, n$)-QSS}

This section presents two frameworks for entanglement-free $(t,n)$-QSS, which allow any classical $(t,n)$-SS to be used to enable the threshold structure. In addition, several single quantum systems can be applied to the frameworks and make QSS schemes be entanglement free.

Let $t$, $n~(t\leq n)$ be two positive integers, then an entanglement-free ($t, n$)-QSS is a quantum version of ($t, n$)-SS.  It divides a secret into $n$ shares and allocates each share to a shareholder such that minimum $t$ shareholders are required to recover the secret, while no entangled states are involved in the protocol.

\subsection{Basic framework of entanglement-free ($t, n$)-QSS}
An entanglement-free ($t, n$)-QSS consists of four algorithms, \emph{classical private Share Distribution} \textbf{SD}($s, {{\mathcal{U}}}$), \emph{Secret Encoding} \textbf{SE}($\varphi, S$),  \emph{Sequential Operation on single quantum system} \textbf{SO}($\varphi_0, {{\mathcal{U}_m}}, \Omega_m, {P_m}$) and \emph{Secret Reconstruction} \textbf{SR}($\varphi_m, {P_m}$).

\textbf{SD}($s, {{\mathcal{U}}}$)-it takes the private value $s$ and the set of $n$ shareholders ${{\mathcal{U}}}$ as input and generates $\Omega $, the set of $n$ shares as output. It is the same as the share distribution in classical $(t,n)$-SS. In this algorithm, a dealer divides the private value $s$ (the secret in classical $(t,n)$-SS) into $n$ shares and allocates each share to a shareholder securely. Note that this algorithm enables the QSS scheme with a ($t,n$) threshold structure and the private value $s$ is not the secret $S$ to be shared in our framework. We enforce no constraint at this step, $i.e.$, any classical $(t,n)$-SS could be applied to the algorithm.

\textbf{SE}($\varphi , S$)-it provides $\varphi_0$ as output, which is obtained through encoding the secret $S$ into a given quantum state $\varphi $. In this algorithm, the dealer encodes the secret $S$ into an initial single quantum state $\varphi $ and thus transforms $\varphi $ into a new state $\varphi_0$. Then, the dealer sends  $\varphi_0$ to a group of shareholders (participants) who would later collaborate to recover the secret.

\textbf{SO}($\varphi_0, {{\mathcal{U}_m}}, \Omega_m, {P_m}$)-it takes the secret-encoded quantum state $\varphi_0 $, the set of $m~(m\geq t)$ participants ${{\mathcal{U}_m}}$, the corresponding share set $\Omega_m $ and the set of $m$ random numbers ${P_m}$ as input. After $m$ sequential unitary operations performed by ${{\mathcal{U}_m}}$, $\varphi_0$ is mapped into $\varphi_m$ as output. In this algorithm, $m$ participants ${{\mathcal{U}_m}}$ perform sequential unitary operations on the received $\varphi_0$ to embed their shares into the state. Specifically, on receiving $\varphi_{i-1}, 1\le i\le m$, each participant ${U_i}\in {{\mathcal{U}_m}}$ computes a component $c_i$ with its share ${s_i}\in\Omega_m $ and embeds $c_i$ with a private random number $p_i\in{P_m}$ into $\varphi_{i-1}$ to generate $\varphi_i$ by performing some unitary operations. As a result, $\varphi_0$ is finally transformed into quantum state $\varphi_m$. Note that for a single quantum system, many unitary operations such as phase shift, generalized Pauli operation, etc., can be utilized here.

\textbf{SR}($\varphi_m , {P_m}$)-it takes $\varphi_m$ and the set of $m$ private random numbers $P_m$ as input while recovering the secret $S$ to be output. In this algorithm, the last participant $U_m\in {\mathcal{U}_m} $ measures $\varphi_m$ and publishes the measurement result. All participants exchange private random numbers in ${P_m}$ and recover the secret $S$ independently from the measurement result.

%\emph{classical private Share Distribution}

\begin{remark}
As mentioned above, this framework is flexible since all classical $(t,n)$-SS schemes can be used in share distribution and multiple unitary operations can be selected to embed private components into target state. Moreover, due to the application of single quantum system, no entanglement is required in the framework. Therefore, schemes under this framework are more practical in experiments.
\end{remark}

\subsection{Verifiable framework of entanglement-free ($t, n$)-QSS}
Obviously, if any participant embeds a wrong component ($e.g.$, using the wrong share) or the quantum system is eavesdropped in the algorithm of \textbf{SO}($\varphi_0 , {{\mathcal{U}_m}}, \Omega_m , {P_m}$), or the last participant publishes a wrong measurement result in \textbf{SR}($\varphi_m , {P_m}$), some participants will recover a wrong secret. As a matter of fact, the above attacks can be effectively thwarted by simply adding a verification quantum state to the framework. The verifiable framework uses $k-1~(k\geq2)$ single quantum states to encode $k-1$ secrets. Moreover, it utilizes an extra single quantum state to encode the verification value,  by which each participant can verify the correctness of the recovered secrets. As a result, the improved framework is able to detect eavesdropping attacks and cheating attacks.

The verifiable framework consists of four algorithms \textbf{SD}($s, {{\mathcal{U}}}$), \textbf{SE}($\{ \varphi_{v},$ $S_v| v=1,2,\dots,k \}$), \textbf{SO}($\{\varphi_{v0}|v=1,2,\dots,k\}, \mathcal{U}_m, \Omega_m , P_m$) and \textbf{SR}($\{\varphi_{vm}|v=1,2,\dots,k\}, {P_m}$).

The verifiable framework shares the same algorithm \textbf{SD}($s$, ${{\mathcal{U}}}$) with the basic one.

\textbf{SE}($\{\varphi_{v}, S_v| v=1,2,\dots,k$\})- In this algorithm, the dealer first encodes $k-1$ secrets $S_1,S_2,\dots,S_{k-1}$ into initial single quantum states $\varphi_{1},\varphi_{2},\dots,\varphi_{k-1}$ respectively. Then, it encodes $S_{k}$ into the initial state $\varphi_{k}$ as the verification such that $S_{k}=f(S_1,S_2,\dots,S_{k-1})$ holds, where $f(.)$ is a verification function. After encoding, these initial states are transformed into $\varphi_{10},\varphi_{20},\dots,\varphi_{k0}$ respectively. Finally, the dealer sends these states to participants for secret reconstruction.

\textbf{SO}($\{\varphi_{v0}|v=1,2,\dots,k\},\mathcal{U}_m,\Omega_m, P_m$)- Suppose there are $m$ participants ${{\mathcal{U}_m}=\{U_1,U_2,\dots,U_m\}}$ want to recover secrets. The first participant $U_1$  $\in {{\mathcal{U}_m}}$ performs unitary operations on the received $\varphi_{v0}$ to embed its component $c_1$, generated from share ${s_1}\in\Omega_m $ and a random number $p_1 \in P_m$, into each state. As a result, $\varphi_{v0}$ are transformed into states $\varphi_{v1}$. The next participant $U_2\in \mathcal{U}_m$ repeats the operations as $U_1$ does. Finally, all states are transformed into $\varphi_{vm},v=1,2,\dots,k$ by the last participant $U_m$.

\textbf{SR}($\{\varphi_{vm}|v=1,2,\dots,k\}, {P_m}$)- In the algorithm, the last participant $U_m$ measures all quantum states $\varphi_{vm},v=1,2,\dots,k$ and publishes measurement results among $\mathcal{U}_m$. All $m$ participants in $\mathcal{U}_m$ mutually exchange their private random numbers in ${P_m}$. Then each participant recovers secrets $S_1,S_2,\dots,S_{k-1}$ together with the verification value $S_{k}$ independently and checks whether
$S_{k}=f(S_1,S_2,\dots,S_{k-1})$ holds. If it is true, all secrets are recovered correctly; otherwise, the recovered results are discarded.
\begin{remark}
The verifiable framework uses an extra single quantum state to enable verification. As a matter of fact, only one single quantum state is enough for the purpose of both encoding and verification. For a given secret $S\in {\rm{GF}}(q)=\{0,1,\dots,q-1\}$, we first encode it into the form $S'={S|H(S)}\in {\rm{GF}}(d)$ as the new secret, where $H(S)$ is some form of verification value of $S$, $e.g.$, the one-way hash value of $S$ or checksum of $S$ in terms of some type of coding rule; the notation $|$ denotes the operation of concatenation, $d$ is a prime larger than $q$. Then, we employ the basic framework to share the new secret $S'$ (suppose that the basic framework shares a secret in GF$(d)$). In secret reconstruction, each participant checks whether the recovered secret has the form of ${S|H(S)}$, if it does, the secret $S$ is correctly recovered; otherwise, the secret is wrong.
\end{remark}

\section{Quantum Fourier transform (QFT)-based entanglement-free $(t,n)$-QSS}
In this section, we first introduce three basic quantum operations, and then present a concrete QFT based $d$-level $(t,n)$-QSS by following the verifiable framework. Finally, the correctness of the scheme is presented.
\subsection{Three quantum operations}
Let us introduce three useful quantum operations, $i.e.$, QFT, Inverse QFT and generalized Pauli operation. They play important roles in our QSS scheme.
\begin{definition} Quantum Fourier transform is a linear operator performed on $d$ orthogonal basis $\left| 0 \right\rangle ,\left| 1 \right\rangle ,\dots,\left| {d - 1} \right\rangle $ in the following way
\begin{equation}
{\rm{QFT}}\left| j \right\rangle  = \frac{1}{{\sqrt d }}\sum\limits_{k = 0}^{d - 1} {{\omega ^{j \cdot k}}} \left| k \right\rangle  ,
\end{equation}
 where $\omega=e^{2\pi i/d}$ is the $d$th root of unity.
\end{definition}

\begin{definition} Inverse Quantum Fourier Transform (IQFT) is the inverse operator of QFT
\begin{equation}
{\rm{QFT}}{}^{ - 1}\left| k \right\rangle  = \frac{1}{{\sqrt d }}\sum\limits_{j = 0}^{d - 1} {{\omega ^{ - k \cdot j}}} \left| j \right\rangle,
\end{equation}
and it is also a linear operator.
\end{definition}

\begin{definition}
On Hilbert space of a $d$-level quantum system, the generalized Pauli operation $U_{m,n}$ is
\begin{equation}
{U_{m,n}} = \sum\limits_{k = 0}^{d - 1} {{\omega ^{n \cdot k}}\left| {k + m} \right\rangle } \left\langle k \right|,
\end{equation}
where $ m,n \in {\rm{GF}}(d) $, $d$ is a prime \cite{Thas2009}. With QFT, the generalized Pauli operation $U_{m,n}$ may complete the following transformation:
\begin{equation}
\begin{aligned}
\label{4}
{U_{m,n}}{\rm{QFT}}\left| j \right\rangle
    &= {U_{m,n}}\frac{1}{{\sqrt d }}\sum\limits_{k = 0}^{d - 1} {{\omega ^{j \cdot k}}\left| k \right\rangle } \\
    &= \frac{1}{{\sqrt d }}\sum\limits_{k = 0}^{d - 1} {{\omega ^{(j + n) \cdot k}}\left| {k + m} \right\rangle } .
\end{aligned}
\end{equation}
\end{definition}
\subsection{The proposed QSS scheme}
According to the verifiable framework, the scheme is decomposed in four algorithms. Then, the dealer and participants can complete the secret sharing task by running the Algorithm 1 to Algorithm 4 sequentially (see Figure.\ref{F1}). To make the scheme clear, we will discuss each algorithm in details.

\emph{classical private Share Distribution}: In this algorithm, the dealer Alice generates and allocates classical private shares to each shareholder Bob$_j, j=1,2,\dots,n$. For each private value, Alice will generate $n$ shares and send one share to each participant. In this way, Alice uses $k$ private values to generate shares and each participant will receive $k$ shares.
%This algorithm enables the proposed scheme with a ($t,n$) threshold structure.

\begin{algorithm}[H]
\caption{classical private Share Distribution}
{\bf Input:}
private values $s_v,v=1,2,\dots,k$;
the set of $n$ shareholders ${\mathcal{U}}=\{$Bo$b_j|j=1,2,\dots,n\}$  with respective
public information $x_j$
\\
{\bf Output:}
$kn$ shares $\Omega_{vj}$
\\
\rule{\textwidth}{0.1mm}

{\bf Steps:}
\begin{algorithmic}[1]

\State Dealer Alice picks up $k$ random polynomials $f_v(x),v=1,2,\dots, k$ of degree at most $t-1$ ($t\leq n$) over finite field GF$(d)$:
    %\begin{equation}
	 $f_v(x) = {a_{v0}} + {a_{v1}}x + \dots + {a_{v(t - 1)}}{x^{t - 1}}\bmod d,$
    %\end{equation}
    where $s_v = {a_{v0}} = f_v(0)$ denote the private values, and all coefficients ${a_{vr}}, r = 0,1,\dots,t - 1$, are uniformly chosen from the finite field GF$(d)$ for prime $d$.

\State Given shareholders ${\mathcal{U}}$, Alice computes $f_v({x_j})$ as the shares of shareholder Bob$_j$ for $j = 1,2,\dots,n$, where ${x_j} \in {\rm{GF}}(d)$ is non-zero number with ${x_j} \ne {x_l}$ for $j \ne l$.

\State $\Omega_{vj}=f_v(x_j)$ are the $kn$ shares and then sent to the corresponding shareholder Bob$_j$ securely by using quantum key distribution.
\end{algorithmic}
\end{algorithm}
\emph{Secret Encoding}: Assume that Alice wants to share $k-1$ secrets $S_u\in {\rm{GF}}(d),u=1,2,\dots,k-1$ among shareholders. In order to establish a verifiable $(t,n)$-QSS, Alice first picks $S_k \in {\rm{GF}}(d)$ (verification value) with $\prod\nolimits_{u=1}^{k-1} {S_u}=S_k\bmod d$ for public prime $d$, and then uses single quantum system to realize the scheme.

\begin{algorithm}[H]
\caption{Secret Encoding}
{\bf Input:}
$k$ initial single qudits $\varphi_{v}= {\left| {{\Psi _v}} \right\rangle },v=1,2,\dots,k$;
$k-1$ secrets $S_u,u=1,2,\dots,k-1$ and a verification value $S_k$
\\
{\bf Output:}
$k$ encoded single qudits $\varphi_{v0}= {\left| {{\Psi _v}} \right\rangle _0},v=1,2,\dots,k$
\\
\rule{\textwidth}{0.1mm}
{\bf Steps:}
\begin{algorithmic}[1]

\State Alice prepares $k$ qudits which are in the state ${\left| {{\Psi _v}} \right\rangle }=\left| 0 \right\rangle$ as input.

\State For each qudit, she performs QFT and the generalized Pauli operation $U_{0,{p_{v0}+q_{v0}}}$ to encode $ {p_{v0}}={S_v}$ and $q_{v0} =d-s_v$ into the qudit.

\State After operations, she gets $k$ qudits $ {\left| {{\Psi _v}} \right\rangle _0}$ as output which will be used for secret reconstruction.
\end{algorithmic}
\end{algorithm}
\emph{Sequential Operation of single quantum system}: To share the secrets $S_u,u=1,2,\dots,k-1$, arbitrary $m$ $(m\geq t)$ participants ${\rm{Bob}}_j,j=1,2,\dots,m$ can cooperate to achieve the goal. In this algorithm, each participant only needs to complete their own operations on the communicated qudits sequentially.
%Here we suppose Bob$_1$ is the first participant who receives the $k$ qudits from Alice.

\begin{algorithm}[H]
\caption{Sequential Operation of single quantum system}
{\bf Input:}
the secret-encoded quantum states $\varphi_{v0}={\left| {{\Psi _v}} \right\rangle _0}$, the set of $m~(m\geq t)$ participants ${\mathcal{U}_m}=\{$Bo$b_j|j=1,2,\dots,m\}$, the corresponding share set $\Omega_{vm}$ and the set of $m$ random numbers ${P_{vm}} \in $GF$(d)$
\\
{\bf Output:}
$k$ encoded single qudits $\varphi_{vm}= {\left| {{\Psi _v}} \right\rangle _m},v=1,2,\dots,k$
\\
\rule{\textwidth}{0.1mm}
{\bf Steps:}
\begin{algorithmic}[1]

\State Suppose Bob$_1$ is the first participant to receive $\varphi_{v0}$. He first prepares $k$ mutually independent private random numbers $p_{v1} \in {P_{vm}}$ and computes the component ${q_{v1}} = {c_{v1}} = f_v({x_1})\prod\nolimits_{r = 2}^m {\frac{{{x_r}}}{{{x_r} - {x_1}}}} \bmod d$ based on his shares. Then he performs the generalized Pauli operations ${U_{p_{v1},{p_{v1}+q_{v1}}}}$ on each ${\left| {{\Psi _v}} \right\rangle _0}$ respectively. As a result, he gets the processed states $ {\left| {{\Psi _v}} \right\rangle _1}$. Then Bob$_1$ delivers the states ${\left| {{\Psi _v}} \right\rangle _1}$ to Bob$_2$.

\State For each of the other participants Bob$_j,j = 2,3,\dots,m$, upon receiving the $k$ qudits, they repeat the same procedure sequentially as Bob$_1$ does. That is, each Bob$_j$ completes operations ${U_{{p_{vj}},{{p_{vj}}+{q_{vj}}}}}$ on ${\left| {{\Psi _v}} \right\rangle _{j - 1}}$ and gets the states ${\left| {{\Psi _v}} \right\rangle _j}$, where ${p_{vj}}\in {P_{vm}}$ and $q_{vj}=c_{vj} = f_v({x_j})\prod\nolimits_{r = 1,r \ne j}^m {\frac{{{x_r}}}{{{x_r} - {x_j}}}} \bmod d$. Subsequently, Bob$_j,j=2,3,\dots,m-1$ delivers the qudits ${\left| {{\Psi _v}} \right\rangle _j}$ to the next participant Bob$_{j + 1}$.

\State After operations, the last one Bob$_m$ gets the $k$ qudits $ {\left| {{\Psi _{v}}} \right\rangle _m}$ as output.
\end{algorithmic}
\end{algorithm}	

\emph{Secret Reconstruction}: Finally, the last participant Bob$_m$ keeps and performs IQFT on the $k$ qudits for accurate measurement. Then, all participants can recover the secrets and verify their correctness.
%\begin{enumerate}[($i$)]

\begin{algorithm}[H]
\caption{Secret Reconstruction}
{\bf Input:}
$k$ encoded single qudits $\varphi_{vm}= {\left| {{\Psi _v}} \right\rangle _m}$;
the set of $m$ private random numbers ${P_{vm}} \in $GF$(d)$
\\
{\bf Output:}
$k-1$ secrets $S_u,u=1,2,\dots,k-1$ and a verification value $S_k$
\\
\rule{\textwidth}{0.1mm}
{\bf Steps:}
\begin{algorithmic}[1]

\State The last participant Bob$_m$ performs IQFT on the $k$ qudits ${\left| {{\Psi _v}} \right\rangle _m}$ and measures them in the computational basis $\{ \left| j \right\rangle ,j = 0,1,\dots,d - 1\} $. Then, he publishes the measurement results ${R_v},v=1,2,\dots,k$ to all participants.

\State Each of $m$ participants exchanges their random numbers (e.g. using the unconditionally
secure bit commitment protocol \cite{Kent2012} or exchanging in a simultaneous way), then they can compute the values $p_{v0}=R_v-\sum_{j = 1}^m {p_{vj}} \bmod d$. Because in a valid run, the measurement results ${R_v}$, the dealer's secrets $p_{v0}$ and participants' private random numbers $p_{vj},j = 1,\dots,m$ satisfy globally
    \begin{equation} \label{5}
    {R_v}  =\sum\nolimits_{j = 0}^m {{p_{vj}}} \bmod d,v = 1,2,\dots,k.
    \end{equation}

\State If the recovered values $p_{10},p_{20},\dots,p_{k0}$ satisfy the following relation
	\begin{equation}
\prod\nolimits_{u=1}^{k-1} {p_{u0}}  = p_{k0}\bmod d,
\end{equation}
they can make sure the secret sharing attempt is not corrupt and thus the recovered secrets and the verification value ${S_v}=p_{v0}$ are true; otherwise they are aware that this round is invalid and thus abort it.

\end{algorithmic}
\end{algorithm}

\begin{remark}
Following the above idea, a qudit in a normalized unknown state $\left| \varphi  \right\rangle  = \sum\nolimits_{j = 0}^{d - 1} {{\alpha _j}} \left| j \right\rangle,$ $\sum\nolimits_{j = 0}^{d - 1} {{{\left| {{\alpha _j}} \right|}^2}}  = 1$ can also be shared among at least $t$ shareholders by substituting the generalized Pauli operation $U_{0,q_j}$ for ${U_{{p_{j}},{{p_{j}}+{q_{j}}}}}$.
\end{remark}

\begin{figure}[htbp]
  \centering
  \includegraphics[width=1.0\textwidth]{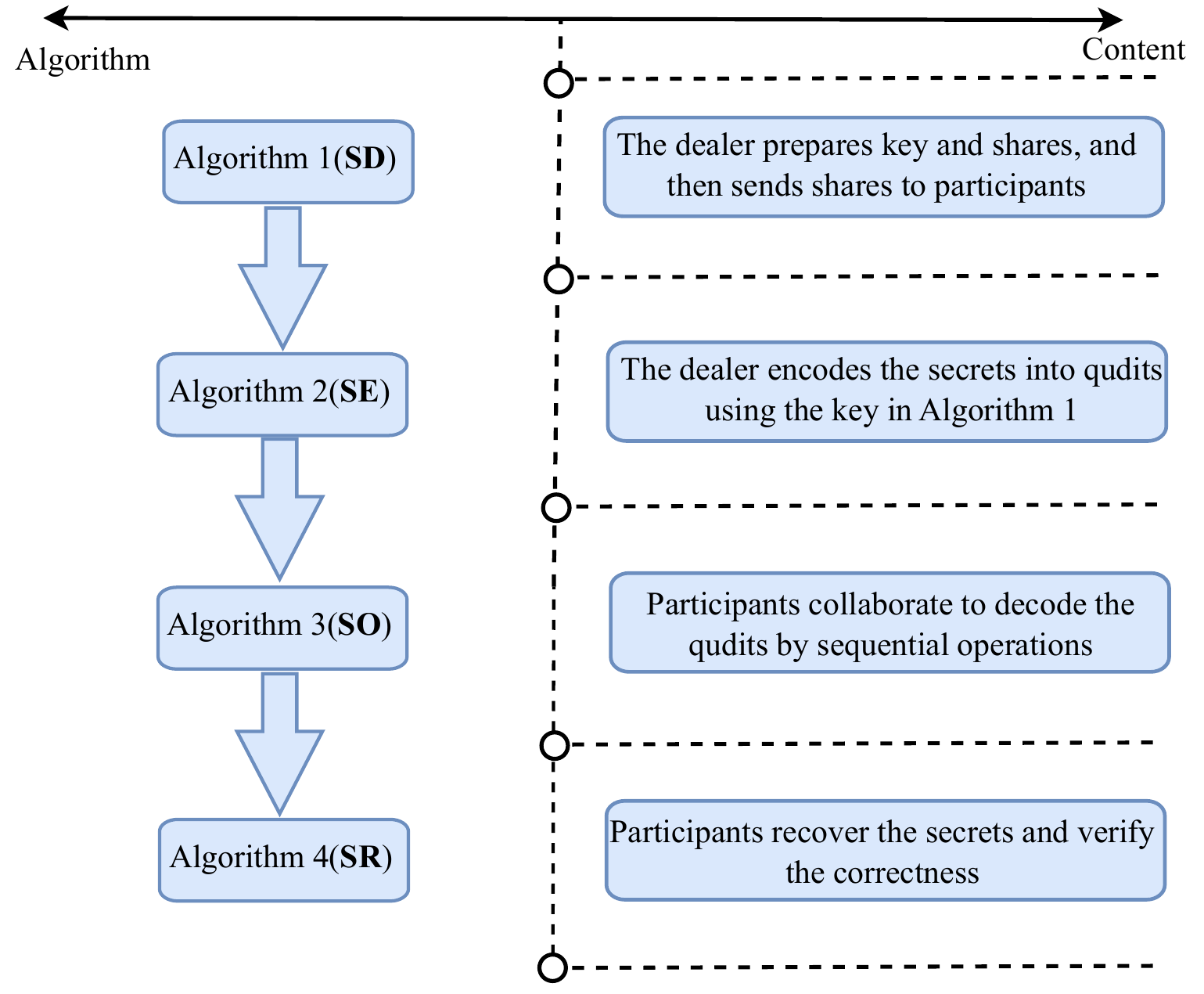}\\
  \caption{Flowchart of our proposed framework}
  \label{F1}
\end{figure}
\subsection{Correctness of the scheme}
Now, we show the correctness of the proposed scheme. The proposed scheme uses $k$ qudits $ {\left| {\Psi_v } \right\rangle _0}, v=1,2,\dots,k$ to share $k-1$ secrets and one verification value. Since all operations on each qudit are similar, we remove the subscript $v$ for simplicity, $e.g.$, uniformly use ${\left| {\Psi} \right\rangle _0}$ to denote each qudit ${\left| {\Psi}_v \right\rangle _0},v=1,2,\dots,k$.

With dealer's operations, qudits are initially prepared in the state $ {\left| {\Psi } \right\rangle _0}=\frac{1}{{\sqrt d }}\sum\nolimits_{k = 0}^{d - 1} {{\omega ^{(p_0+ q_0) \cdot k}}\left| {k } \right\rangle }$.  After all $m$ participants sequentially complete their operations, the final state becomes
\begin{equation} \label{7}
\begin{aligned}
{\left| \Psi  \right\rangle _m}
    &=\left( {\prod\limits_{j = 1}^m {{U_{{p_j},{p_j} + {q_j}}}} } \right){\left| \Psi  \right\rangle _0}\\
    &=\frac{\xi_m}{{\sqrt d }}\sum\limits_{k = 0}^{d - 1} {{\omega ^{\left( {\sum\nolimits_{j = 0}^m {({p_j} + {q_j})} } \right) \cdot k}}} \left| k \right\rangle \\
    &=\frac{\xi_m}{{\sqrt d }}\sum\limits_{k = 0}^{d - 1} {{\omega ^{(\sum\nolimits_{j = 0}^m {{p_j}} {\rm{ + }}d - s + \sum\nolimits_{j = 1}^m {{c_j}} ) \cdot k}}\left| k \right\rangle } \\
    &=\frac{\xi_m}{{\sqrt d }}\sum\limits_{k = 0}^{d - 1} {{\omega ^{(\sum\nolimits_{j = 0}^m {{p_j}}+L\cdot d ) \cdot k}}\left| k \right\rangle }, (L\in Z)
\end{aligned}
\end{equation}
where $\xi_m={\omega ^{ - \sum\nolimits_{a = 1}^m {{p_a}} \left( {\sum\nolimits_{b = 0}^{a} {({p_b} + {q_b})} } \right)}}$ is an overall phase term (the more detailed proof see the Appendix.\ref{B}). Since in Shamir's ($t ,n$)-SS, any participant can recover the secret by summing up all components, $i.e.$,
\begin{equation}
\begin{aligned}
s=\sum_{j=1}^{m}{c_j} \bmod d=\sum_{j=1}^{m}{f({x_j})\prod\limits_{r = 1,r \ne j}^m {\frac{{{x_r}}}{{{x_r} - {x_j}}}} \bmod d}.
\end{aligned}
\end{equation}
Thus we have $\sum_{j = 1}^m {{c_j}}=Nd+s,N \in Z$ to ensure the correctness of the final state in Eq.(\ref{7}). Then Bob$_m$ performs IQFT on the qudit, giving
\begin{equation}
{\rm{QFT}}{}^{ - 1}{\left| \Psi  \right\rangle _m} = \left| {\sum\nolimits_{j = 0}^m {{p_j}} \bmod d} \right\rangle .
\end{equation}
After measuring the state in the computational basis, Bob$_m$ can get the results satisfying Eq.(\ref{5}). By exchanging random numbers, all participants can recover the secrets from the measurements as shown in \emph{Secret Reconstruction-Step2}. Note that, in the last step, we drop the phase term $\xi_m$ since a global phase shift does not affect the (inverse) Fourier components of a given state.

To illustrate the feasibility of the scheme, we hereby give an example of (4, 6) threshold quantum secret sharing, which shares a classical secret by a single qudit, as follows.

During the \emph{classical private Share Distribution}, the dealer Alice first chooses a random polynomial $f(x)$ of degree 3 over GF$(23)$: $f(x) = 17 + 5x + 12{x^2} + 6{x^3}\bmod 23$, and thus the private value is $s={a_0} = f(0) = 17$ with threshold $t = 4$ and the prime $d = 23$. Then she computes and allocates a share $f({x_j})$ to each shareholder Bob$_j$ with public information ${x_j} = j + 1$ for $j = 1,2,\dots,6$. As a result, $f({x_1} = 2) = 123\bmod 23 = 8$, $f({x_2} = 3) = 302\bmod 23 = 3$, $f({x_3} = 4) = 15$, $f({x_4} = 5) = 11$, $f({x_5} = 6) = 4$ and $f({x_6} = 7) = 7$.

Then, in the \emph{Secret Encoding} algorithm, Alice shares the secret $S=14$ among the participants Bob$_j$, $j=1,2,\dots,6$. She first performs QFT on the state $| \Psi \rangle = | 0 \rangle$ and generalized Pauli operation $U_{0,{p_{0}+q_{0}}}$ to encode $ {p_{0}}=S=14$ and $q_{0} =d-s=6$ into the qudit. Finally she gets the qudit $| \Psi \rangle _ { 0 }=\frac { 1 } { \sqrt { 23 } } \sum _ { k = 0 } ^ { 22 } \omega ^ {  20 \cdot k } | k \rangle$.

In the algorithm of \emph{Sequential Operation of single quantum system}, suppose participants Bob$_j$ with $j\in \mathcal{U}_4=\{1,3,4,6\}$, want to reconstruct the secret, they each prepare a random number, such as ${p_1}=0,{p_3}=3,{p_4}=16,{p_6}=9$ and compute a component by Lagrange interpolation as follows:
${q_1}={c_1} = f({x_1})\prod\limits_{r \in \mathcal{U}_4, r \ne 1} {\frac{{{x_r}}}{{{x_r} - {x_1}}}} \bmod p = f(2) \cdot \frac{4}{{4 - 2}} \cdot \frac{5}{{5 - 2}} \cdot \frac{7}{{7 - 2}}\bmod 23 = 22,$
${q_3}={c_3} = 15 \cdot \frac{2}{-2} \cdot \frac{5}{1} \cdot \frac{7}{3}\bmod 23 = 9,$
${q_4}={c_4} = 11 \cdot \frac{2}{{-3}} \cdot \frac{4}{{-1}} \cdot \frac{7}{{2}}\bmod 23 = 3,$
${q_6}={c_6} = 7 \cdot \frac{2}{{-5}} \cdot \frac{4}{{-3}} \cdot \frac{5}{{-2}}\bmod 23 = 6.$
Then they sequentially perform the generalized Pauli operations $U_{0,22},U_{3,12},U_{16,19},U_{9,15}$ on the processed qudit received from Alice. At last the last participant Bob$_6$ keeps the qudit
\begin{equation}
\begin{aligned}
{\left| \Psi  \right\rangle _4}
    &=\left( {\prod\limits_{j = 1}^4 {{U_{{p_j},{p_j} + {q_j}}}} } \right){\left| \Psi  \right\rangle _0}\\
    &=\frac{\xi_4}{{\sqrt 23 }}\sum\limits_{k = 0}^{22} {{\omega ^{\left( {\sum\nolimits_{j = 0}^4 {({p_j} + {q_j})} } \right) \cdot k}}} \left| k \right\rangle \\
    &=\frac{\xi_4}{{\sqrt 23 }}\sum\limits_{k = 0}^{22} {{\omega ^{(14+0+3+16+9+6+22+9+3+6) \cdot k}}\left| k \right\rangle } \\
    &=\frac{\xi_4}{{\sqrt 23 }}\sum\limits_{k = 0}^{22} {{\omega ^{88 \cdot k}}\left| k \right\rangle },
\end{aligned}
\end{equation}
where $\xi_4$ is an overall phase term.

In the \emph{Secret Reconstruction}, the last participant Bob$_6$ performs IQFT on the qudit ${\left| \Psi  \right\rangle _4}$ and measures it in the computational basis. After that he can get the result $R=19$ because of ${\rm{QFT}}{}^{ - 1}{\left| \Psi  \right\rangle _4} = \left| 88 \bmod 23 \right\rangle= \left| 19 \right\rangle$ and publishes it. Finally, each of these four participants exchanges their random numbers, the secret $S=19-0-3-16-9\bmod 23=14$ can be recovered.

\section{Security analysis}
Schemes for secret sharing have to guarantee security. But, almost two decades after the first QSS, there exists no such a scheme (with or without entanglement) which has been proven to be unconditionally secure against cheating of dishonest users. A recent work \cite{Kogias2017} presents a feasible entanglement-based continuous variable QSS scheme. It derives sufficient conditions for providing unconditional security of the dealer's classical secret against general attacks of an eavesdropper and arbitrary cheating strategies. Furthermore, the work's results pave the way for experimental demonstration of an unconditionally secure QSS.

In this paper, we employ Shamir's ($t, n$)-SS to support the threshold structure and provide information theoretical security. In \emph{Secret Encoding}, the dealer Alice adds each secret (including the verification value) $S_v={p_{v0}},v=1,2,\dots,k$ to the private value $s_v$ by the generalized Pauli operation $U_{0,{p_{v0}}+{q_{v0}}}$. After the last participant ${\rm{Bob}}_m$ completes his Pauli operation and IQFT, each measurement result is $R_v=\sum_{j=0}^{m}p_{vj}+d-s_v+\sum_{j=1}^{m}{c_{vj}}\bmod d, v=1,2,\dots,k$. This process can be considered as an encryption of the secret ${p_{v0}}$. Obviously, with the published $R_v$, a participant can reconstruct ${p_{v0}}$ only after collecting all private random numbers $p_{vj}, j=1,2,\dots, m$ and all $m$ components $c_{vj}$ correctly. As a result, many attacks do not work, including intercept-resend attack and entangle-and-measure attack mounted by an external eavesdropper or participant attack in association with entanglement swapping \cite{Gao2007,He2007}. These attacks' more detailed analysis are mostly similar to works \cite{Lu2018,Qin2015}.

In the following, we use \textbf{Theorem} \ref{T1} to prove the security against the collusion attack by less than $t$ participants. If less than $t$ participants obtain no information about the secret, then a $(t,n)$-QSS scheme can be concluded to be perfect with respect to the probability distribution of secret over secret space.

\begin{theorem} \label{T1}The proposed $(t,n)$-QSS scheme is perfect with respect to the probability distribution of secret over secret space. That is,
\begin{equation}
I(S_v;{\Omega}) = H(S_v) - H(S_v|{\Omega}) =0,
\end{equation}
where $H(S_v)$ is the information entropy of the secret $S_v,v=1,2,\dots,k$, ${\Omega}$ denotes the set of shares available for less than $t$ participants and $I(S_v;{\Omega})$ represents the mutual information of $S_v$ with ${\Omega}$.
\end{theorem}

\begin{proof}
In normal case, $m(m\ge{t}) $ participants $\{{\rm{Bob}}_j, j=1,2,\dots,m\}$, with the corresponding shares $\{f_v(x_j),v=1,2,\dots,k\}$, can cooperate to recover all secrets $S_v$.

Without losing generality, suppose exactly $t-1$ participants $\{{\rm{Bob}}_j, j=1,2,\dots,t-1\}$ conspire to achieve the secrets. At first, the participant ${\rm{Bob}}_{t-1}$ measures ${\left| {{\Psi _v}} \right\rangle _{t-1}}$ and publishes the results $R_v=\sum_{j=0}^{t-1}p_{vj}+d-s_v+\sum_{j=1}^{t-1}{c_{vj}}\bmod d$. After exchanging private random numbers and shares $\{p_{vj}, f_v(x_j),j=1,2,\dots,t-1\} $, each participant ${\rm{Bob}}_j$ obtains the following results
${R_v}'={p_{v0}}+d-s_v\bmod d$.

Since from the view of participants, each of the secrets ${p_{v0}}$ selected from GF($d$) by Alice is a uniform and random variable, thus, $S_v={p_{v0}}$ are indistinguishable from a uniformly distributed random variable over GF$(d)$, $i.e.$, $P(S_v)=1/d$. As a result, the entropy of $S_v$ is $H(S_v)=\log d$.

Just like Shamir's $(t,n)$-SS \cite{Shamir1979}, for $m\ge{t}$, each private value can be computed by Lagrange's interpolation
\begin{equation}
s_v=\sum_{j=1}^{m}{f_v({x_j})\prod\limits_{r = 1,r \ne j}^m {\frac{{{x_r}}}{{{x_r} - {x_j}}}} \bmod d}.
\end{equation}
So each $s_v$ is a random variable uniformly distributed over GF$(d)$ with less than $t$ shares. In other words, $t-1$ participants $\{{\rm{Bob}}_j, j=1,2,\dots,t-1\}$ conspire with the shares $\Omega=\{f_v(x_j),j=1,2,\dots,t-1\}$ available, they can obtain $s_v$ only with the probability $P(s_v|\Omega)=1/d$, $i.e.$, $P(S_v|\Omega)=1/d$ due to ${R_v}'={p_{v0}}+d-s_v\bmod d$ and $S_v=p_{v0}$. Consequently, the conditional entropy can be computed as $H(S_v|\Omega)=\log d$.

In conclusion, we finally have
\begin{equation}\label{13}
I(S_v;{\Omega}) = H(S_v) - H(S_v|{\Omega}) =0.
\end{equation}

Because of the $k-1$ secrets $S_u,u=1,2,\dots,k-1$ and the verification value $S_k$ all satisfying the above Eq.(\ref{13}), thus with respect to probability distribution of secrets in the secret space GF$(d)$, we can conclude that the proposed scheme is perfect.
\end{proof}

To free the trusted third party, the last participant is appointed to measure the qudits and thus, he directly knows the measurement results ($i.e.$, the summations of each secret and random numbers). By publishing fake measurement results, the last participant itself recovers true secrets while making others obtain wrong secrets. Of course, other participant can also cheat the rest ones by using a fake random number in the secret reconstruction. Moreover, the qudits are obviously vulnerable to be eavesdropped. Thus the proposed scheme establishes a verification mechanism to detect such cheating or eavesdropping.

Here we consider the error rate of the verification mechanism, which is the probability that the verification mechanism does not detect wrong secrets. The ideal error rate is certainly 0 since it means all false secrets can be detected during secret reconstruction.
\begin{theorem}In the proposed scheme, the error rate of the verification mechanism converges to 0 when the dimension of the secret approaches to infinity. That is,
\begin{equation}
	 \mathop {\lim }\limits_{d \to \infty }{r_e=0},
\end{equation}
where $r_e$ denotes the error rate and $d$ represents the dimension of the secret.
\end{theorem}

\begin{proof}
After the measurement, the last participant ${\rm{Bob}}_m$ can get the correct measurement results ${R_v}$, but cannot know the value about the sum of other participants' random numbers, which is a constant variable ${N_v=\sum_{j=1}^{m}{p_{vj}}}, v=1,2,\dots,k$, known to him. Assume that the last participant ${\rm{Bob}}_m$ publishes the wrong measurements $R_v^\prime \neq{R_v}$ to cheat others. Obviously, if $\prod\nolimits_{u=1}^{k-1} ({R_u^\prime-{N_u}}) =R_k^\prime-{N_k}\bmod d$, happens to hold, then the wrong measurement results will convince other participants and cannot be detected. In this case, the verification mechanism fails to detect the cheating and only the cheater can recover the true secrets while others cannot.
%\[t_{k - 1}^r\]
To be specific, if ${\rm{Bob}}_m$ chooses $k$ values $R_v^\prime,v=1,2,\dots,k $ in GF$(d)$ randomly and uniformly as measurement results and publishes them to the other participants, thus there will be totally $d^k$ tuples of  $\{R_1^\prime, R_2^\prime,\dots,R_k^\prime\}$. Note that ${R_v}, v=1,2,\dots,k$ are published before all participants exchange their random numbers ${p_{vj}},j=1,2,\dots,m$ to achieve the sums ${N_v}$. Since each participant ${\rm{Bob}}_j$ privately and independently chooses his random numbers ${p_{vj}}$, thus in the view of all participants, ${N_v=\sum_{j=1}^{m}{p_{vj}}}$ are indistinguishable from random numbers uniformly distributed in GF$(d)$. As a result, it is same for $(R_v^\prime-{N_v})$, which are also indistinguishable from random numbers uniformly distributed in GF$(d)$ for participants. In this case, since given $\{N_1,N_2,\dots,N_k\}, $ $R_k^{\prime}$ can always be determined for randomly selected sets of $\{R_1^\prime, R_2^\prime,\dots,R_{k-1}^{\prime}\}$, so there are totally $d^{k-1}$ randomly selected tuples of $\{R_1^\prime, R_2^\prime,\dots,R_k^\prime\}$ satisfying
\begin{equation}
	\prod\nolimits_{u=1}^{k-1} ({R_u^\prime-{N_u}}) =R_k^\prime-{N_k}\bmod d.
\end{equation}

The result is the same if any other participant, $i.e.$, Bob$_j,j=1,2,\dots,m-1$ cheats by releasing any different random number ${p_{vj}^\prime}\ne {p_{vj}}, v =1,2,\dots,k$ when exchanging random numbers.

Therefore, the error rate of the verification mechanism is $r_e=d^{k-1}/d^k=1/d$. Thus,
\begin{equation}
	\mathop {\lim }\limits_{d \to \infty }{r_e}=\mathop {\lim }\limits_{d \to \infty }{\frac{d^{k-1}}{d^{k}}}=\mathop {\lim }\limits_{d \to \infty }{\frac{1}{d}}=0. 	
\end{equation}
That is, when $d$ approaches to infinity the error rate will converge to 0.

To sum up, the verifiable mechanism of the scheme can detect the cheating by each participant with the probability $(d-1)/d$, which converges to $100\%$ if $d$ is larger enough.
\end{proof}

\section{Comparisons and discussion}

\subsection{Related work and comparisons}
%\centering{\tablenote {Shared secrets are quantum state(Q) and classical message(C).}}
Since the proposal of the first QSS \cite{Hillery1999}, various extension schemes have been proposed in last two decades. Many of them are based on entangled states (HBB-type). Due to the high cost of engineering multiparticle entangled states, efforts have been made for more economical HBB-type QSS through reducing the number of required particles \cite{Tittel2001,Deng2005}. In a work \cite{Yu2008}, the authors further generalize HBB-type QSS to $d$-level platform by utilizing multiparticle ($>3$) entangled GHZ states. An interesting entanglement-based QSS using entangled state as the secure carriers and splitters of information has been studied in \cite{Bagherinezhad2003}. However, all those schemes are poor in scalability with growing participants and may easily render an invalid run because a participant may fail in measurement due to inefficient detection. Different from entanglement-based HBB-type QSS, some entanglement-free schemes \cite{Schmid2005,Tavakoli2015,Karimipour2015} have also been proposed. But these schemes will be of less interest for secret sharing due to some drawbacks, $e.g.$ they are all ($n, n$) structures which are not flexible under different applications and they offer no unconditional security.

Considering QSS schemes with ($t, n$) structure, the first one was proposed in 1999 \cite{Cleve1999}. It shows that the only constraint on the existence of ($t, n$) threshold schemes comes from quantum no-cloning theorem, which requires $n<2t$. However, the coding process given in this paper, although efficient, is difficult to implement. It's also hard to extend with fixed mapping rules. Later, some other schemes with general ($t, n$) threshold structure were proposed. Among these schemes, \cite{Lance2003,Lau2013} benefit from continuous variable and thus, they are easier to be implemented in practical experiments. Others employ graph states, which provide a superb resource for secret sharing, to construct QSS schemes \cite{Markham2008,Keet2010}. Recently, a new method was developed by taking advantage of the ability of exactly distinguishing orthogonal multipartite entangled states under restricted local operation and classical communication \cite{Rahaman2015,Wang2017}. Besides certain special quantum systems, classical ($t, n$)-SSs can also be incorporated to support the threshold structure in QSS schemes \cite{Tokunaga2005,Qin2015,Lu20182}. But in those schemes, a trusted third party ($e.g.$, the dealer) is required to measure the quantum states.

Thinking about the qubit efficiency, we use ${\eta _q} = \frac{c}{q}$ to denote it, where $c$ is the number of shared classical secret's bits and $q$ is the number of qubits used in the transmission and eavesdropping checking. It can be seen in our scheme, $k$ qudits are used to share $k-1$ secrets in GF$(d)$, thus ${\eta _q} =  \frac{{(k-1)\log d}}{{k\log d}} = \frac{k-1}{k}$. In \cite{Hillery1999}, the prepared 3-qubit GHZ can only establish a bit joint secret, then ${\eta _q}$ is $1/3$. In Schmid's scheme \cite{Schmid2005}, only a qubit is used to sequentially communicate with $n$ users, and one bit is shared, thus ${\eta _q}$ is 1. In the paper \cite{Rahaman2015}, to share a bit classical secret, a pair of distance-$r(>0)$ orthogonal $n$-qubit Dicke states are used, obviously the qubit efficiency ${\eta _q}$ is $1/2n$.

Compared with some previous schemes in Table \ref{table 1}, our scheme can stand out because it only employs sequential $d$-level unitary operations in association with classical ($t,n$) threshold secret sharing on a single qudit. It shows the great scalability and possesses strict threshold structure with information theoretical security. Furthermore, benefiting from the verification mechanism, the proposed ($t, n$)-QSS scheme no longer requires any trusted third party responsible for measurement results, and any cheating strategy of each participant or eavesdropping can be detected.
%p{41pt}<{\centering}
\begin{table*}
\scriptsize
\centering
\resizebox{\textwidth}{!}{
\begin{threeparttable}
	
    \caption{Comparisons between previous QSSs and ours.}\label{table 1}
	
	\begin{tabular}{l|c|c|c|c|c|c}
		\hline
         Schemes &Our scheme &Ref.\cite{Shamir1979} &Ref.\cite{Cleve1999}&Ref.\cite{Hillery1999}&Ref.\cite{Schmid2005}&Ref.\cite{Rahaman2015}\\
        \hline
        Initial state & single qudit & &single qudit& \tabincell{c}{3-qubit\\ GHZ state} &single qubit & \tabincell{c}{$n$-qubit \\Dicke state}\\ \hline
        Shared secrets$^1$ &both  &C&Q&C&C&C \\ \hline
        ($t, n$) threshold & Yes &Yes&Yes&No&No&Yes\\ \hline
        Qubit efficiency & ($k$-1)/$k$ & & &1/3 &1&1/(2$n$)\\
        \hline
        Cheat detection &Yes& No &No&No&No&Yes\\ \hline
	\end{tabular}

\begin{tablenotes}
\item[1] Shared secrets are quantum state(Q) and classical message(C).
\end{tablenotes}
\end{threeparttable}}

\end{table*}

\subsection{Discussion}
For the experimental aspects, these two works \cite{Schmid2005,Hai2013} implemented a quantum secret sharing scheme based on a single qubit (a two-level system). Our $d$-level QSS scheme can also be realized using the same technologies if $d=2$. However, qudits with their state space of dimension $d > 2$, open fascinating experimental prospects. The quantum properties of their states provide new potentialities for quantum information, quantum contextuality, expressions of geometric phases, facets of quantum entanglement and many other foundational aspects of the quantum world that are unapproachable via qubits \cite{Godfrin2018}. So far, the experimental implementations of preparing and manipulating qudit states for large $d$ are still under active investigation. For example, in \cite{Godfrin2018}, they have experimentally investigated the quantum dynamics of a qudit ($d=4$) that consists of a single 3/2 nuclear spin embedded in a molecular magnet transistor geometry, coherently driven by a microwave electric field. In the paper \cite{Giordani2019}, they affirm the potential that one-dimensional quantum walk dynamics represents a valid tool in the task of engineering arbitrary quantum states in a linear-optics platform. Moreover, confirming the feasibility of the protocol by preparing and measuring different classes of relevant qudit states in a six-dimensional space. So, with the above technologies becoming riper, our scheme can be experimented in the arbitrary $d$-dimension.

\section{Conclusion}
This work proposes a verifiable framework for threshold QSS without entanglement. Such a scheme enables QSS without entanglement by incorporating any single quantum system and utilizes existing classical $(t,n)$-SS to keep $(t,n)$-threshold structure. Besides, a verification mechanism is established for thwarting cheating and eavesdropping attacks. As an example, we demonstrate a concrete $(t,n)$-QSS scheme using QFT and the generalized Pauli operation. It shares $k-1$ secrets and an extra verification value with sequential applications of the generalized Pauli operation. At last, each participant can independently recover secrets and verify the correctness from measurement results. We further prove that the proposed scheme is information theoretically secure and the verification mechanism is sufficient since the error rate converges to 0 if the dimension of the secret approaches infinity. Therefore, this class of entanglement-free ($t, n$)-QSS schemes constructed under our framework can address the drawbacks in previous entanglement-free QSS and will be more useful in quantum communication networks due to their intrinsic scalability, flexibility and information theoretical security.

\section*{Acknowledgments}
We would like to thank the anonymous reviewers for helpful suggestions.
This work is supported by the National Natural Science Foundation of China under Grant Nos. 61572454, 61572453,
61520106007 and Anhui Initiative in Quantum Information Technologies under Grant No. AHY150100.

\appendix
\section{List of abbreviations}\label{A}
In this section, we list the full descriptions for the most frequently used abbreviations in the main text.
  \begin{table}[H]
  \centering
  \scriptsize
  \caption{List of useful abbreviations.}\label{table A1}
  \begin{threeparttable}
	
    %\caption{The abbreviations of the article.}\label{table A1}
	
	\begin{tabular}{c|c}
		\hline
         Abbreviation & Description \\
        \hline
        QSS & Quantum secret sharing\\
        \hline
        HBB & Hillery-Buzek-Berthiaume \\
        \hline
        GHZ& Greenberger-Horne-Zeilinger \\
        \hline
        SD & classical private Share Distribution\\
        \hline
        SE &Secret Encoding\\
        \hline
        SO &Sequential Operation of single quantum system\\
        \hline
        SR &Secret Reconstruction\\
        \hline
        GF &Galois Field\\
           &GF($d$)=$\{0,1,2,\dots,d-1\}$ \\
        \hline
        QFT &Quantum Fourier transform\\
        \hline
        IQFT,QFT$^{-1}$ &Inverse Quantum Fourier Transform\\
        \hline
	\end{tabular}
\end{threeparttable}
  %\caption{}\label{}
\end{table}

\section{Proof of Eq.(\ref{7})}\label{B}
  In the paper, with $\omega=e^{2\pi i/d}$ we can first prove \[\sum\nolimits_{j = 0}^{d - 1} {{\omega ^{rj}}}  = \left\{ {\begin{array}{*{20}{c}}
{d,r = 0\bmod d}\\
{0,r \ne 0\bmod d.}
\end{array}} \right.\]
\begin{proof} We have $\omega ^ { d } = e ^ { 2 \pi i } = \cos ( 2 \pi ) + i \sin ( 2 \pi ) = 1$. At first, we suppose $r = 0\bmod d$, thus \[\sum\nolimits_{j = 0}^{d - 1} {{\omega ^{rj}}}  = \sum\nolimits_{j = 0}^{d - 1} {{\omega ^{kdj}}}  = \sum\nolimits_{j = 0}^{d - 1} 1  = d(k \in Z).\]
\par If $r \ne 0\bmod d$, so ${\omega ^r} \ne 1$. Therefore, by using the sum of geometric series, we can get \[\sum\nolimits_{j = 0}^{d - 1} {{\omega ^{rj}}}  = 1 + {\omega ^r} + {\omega ^{2r}} +  \cdots  + {\omega ^{(d - 1)r}} = \frac{{1 - {\omega ^{dr}}}}{{1 - {\omega ^r}}} = \frac{{1 - 1}}{{1 - {\omega ^r}}} = 0.\]
This completes the proof.
\end{proof}
\par Next we can define $\left| {{\mu _j}} \right\rangle  = {\rm QFT}\left| j \right\rangle  = \frac{1}{{\sqrt d }}\sum\limits_{k = 0}^{d - 1} {{\omega ^{jk}}} \left| k \right\rangle,j=0,1,\dots,d-1 $. Moreover, consider the
generalized Pauli operators $X : = \sum _ { k = 0 } ^ { d - 1 } | k + 1 \rangle \langle k |$ and $Z:=\sum _ { k = 0 } ^ { d - 1 } \omega ^ { k } | k \rangle \langle k |$. After performing these two operators on the state $\left| {{\mu _j}} \right\rangle$ we have \[X\left| {{\mu _j}} \right\rangle  = {\omega ^{ - j}}\left| {{\mu _j}} \right\rangle ,Z\left| {{\mu _j}} \right\rangle  = \left| {{\mu _{j + 1}}} \right\rangle .\]
Here we proof the transformation of the generalized Pauli operator $X$.\\
\begin{proof}
\begin{equation*}
\begin{aligned}
X\left| {{\mu _j}} \right\rangle  &= \frac{1}{{\sqrt d }}\sum\limits_{k = 0}^{d - 1} {{\omega ^{jk}}} \left| {k + 1} \right\rangle  = {\omega ^{ - j}}\frac{1}{{\sqrt d }}\sum\limits_{k = 0}^{d - 1} {{\omega ^{j(k + 1)}}} \left| {k + 1} \right\rangle\\
&\mathop  = \limits^{k + 1 = r} {\omega ^{ - j}}\frac{1}{{\sqrt d }}(\sum\limits_{r = 1}^{d - 1} {{\omega ^{jr}}} \left| r \right\rangle  + \left| 0 \right\rangle ) = {\omega ^{ - j}}\sum\limits_{r = 0}^{d - 1} {{\omega ^{jr}}} \left| r \right\rangle  = {\omega ^{ - j}}\left| {{\mu _j}} \right\rangle.
\end{aligned}
\end{equation*}
\end{proof}
As the definition in the paper, the generalized Pauli operation $U_{m,n}$ is
\begin{equation*}
{U_{m,n}} = \sum\limits_{k = 0}^{d - 1} {{\omega ^{n \cdot k}}\left| {k + m} \right\rangle } \left\langle k \right|,
\end{equation*}
where $ m,n \in {\rm{GF}}(d) $. Moreover, it can be written as ${U_{m,n}}=X^mZ^n$. Because with ${X^m} = \sum\nolimits_{k = 0}^{d - 1} {\left| {k + m} \right\rangle } \left\langle k \right|,{Z^n} = \sum\nolimits_{k = 0}^{d - 1} {{\omega ^{nk}}\left| k \right\rangle } \left\langle k \right|$,
we have
\begin{equation*}
\begin{aligned}
{U_{m,n}} &= \sum\limits_{k = 0}^{d - 1} {{\omega ^{n \cdot k}}\left| {k + m} \right\rangle } \left\langle k \right|\\
&=\left( {\sum\nolimits_{k = 0}^{d - 1} {\left| {k + m} \right\rangle } \left\langle k \right|} \right)\left( {\sum\nolimits_{k = 0}^{d - 1} {{\omega ^{nk}}\left| k \right\rangle } \left\langle k \right|} \right) = {X^m}{Z^n}.
\end{aligned}
\end{equation*}
So, the Eq.(\ref{4}) in the paper can be rewritten as
\begin{equation*}
\begin{aligned}
{U_{m,n}}{\rm{QFT}}\left| j \right\rangle
    &= {U_{m,n}}\frac{1}{{\sqrt d }}\sum\limits_{k = 0}^{d - 1} {{\omega ^{j \cdot k}}\left| k \right\rangle } \\
    &= {U_{m,n}}\left| {{\mu _j}} \right\rangle=  X^mZ^n\left| {{\mu _j}} \right\rangle= \omega ^{-m(j+n)}\left| {{\mu _{j+n}}} \right\rangle.
\end{aligned}
\end{equation*}
Therefore, we finally give the proof of the Eq.(\ref{7})\\
\begin{proof}
\begin{equation*}
\begin{aligned}
{\left| \Psi  \right\rangle _m}
    &=\left( {\prod\limits_{j = 1}^m {{U_{{p_j},{p_j} + {q_j}}}} } \right){\left| \Psi  \right\rangle _0}=\prod\limits_{j = 1}^m {{X^{{p_j}}}{Z^{{p_j} + {q_j}}}} \left| {{\mu _{{p_0} + {q_0}}}} \right\rangle \\
    &=\omega ^{-p_1(p_0+q_0+p_1+q_1)}\prod\limits_{j = 2}^m {{X^{{p_j}}}{Z^{{p_j} + {q_j}}}} \left| {{\mu _{{p_0} + {q_0}+{p_1} + {q_1}}}} \right\rangle \\
    &=\xi_m\left| {{\mu _{\sum\nolimits_{j = 0}^m {({p_j} + {q_j})} }}} \right\rangle\\
    &=\frac{\xi_m}{{\sqrt d }}\sum\limits_{k = 0}^{d - 1} {{\omega ^{\left( {\sum\nolimits_{j = 0}^m {({p_j} + {q_j})} } \right) \cdot k}}} \left| k \right\rangle \\
    &=\frac{\xi_m}{{\sqrt d }}\sum\limits_{k = 0}^{d - 1} {{\omega ^{(\sum\nolimits_{j = 0}^m {{p_j}} {\rm{ + }}d - s + \sum\nolimits_{j = 1}^m {{c_j}} ) \cdot k}}\left| k \right\rangle } \\
    &=\frac{\xi_m}{{\sqrt d }}\sum\limits_{k = 0}^{d - 1} {{\omega ^{(\sum\nolimits_{j = 0}^m {{p_j}}+L\cdot d ) \cdot k}}\left| k \right\rangle }, (L\in Z)
\end{aligned}
%\begin{aligned}
%| \Psi \rangle _ { 1 } &= U _ { p _ { 1 } , p _ { 1 } + q _ { 1 } } | \Psi \rangle _ { 0 }\\
 %   &= \frac { 1 } { \sqrt { d } } U _ { p _ { 1 } , p _ { 1 } + q _ { 1 } } ( \sum _ { k = 0 } ^ { d - 1 } \omega ^ { \left( p _ { 0 } + q _ { 0 } \right) k } ) | k \rangle\\
 %   &= U _ { p _ { 1 } , p _ { 1 } + q _ { 1 } }\left| {{\mu _{p_0+q_0}}} \right\rangle\\
 %   &=  X^{p_1}Z^{p_1+q_1}\left| {{\mu _{p_0+q_0}}} \right\rangle\\
  %  &= \omega ^{-p_1(p_0+q_0+p_1+q_1)}\left| {{\mu _{p_0+q_0+p_1+q_1}}} \right\rangle.
%\end{aligned}
\end{equation*}
with the overall phase term $\xi_m={\omega ^{ - \sum\nolimits_{a = 1}^m {{p_a}} \left( {\sum\nolimits_{b = 0}^{a} {({p_b} + {q_b})} } \right)}}$.
\end{proof}

%\section*{References}

%\bibliographystyle{unsrt}
%\bibliography{references}

\end{document}